\documentclass[preprint,natbib]{sigplanconf}

\usepackage{amsthm}
\usepackage{amsmath}
\usepackage{amssymb} % for mathbb
\usepackage{alltt}
\usepackage[pdftex]{graphicx}
\usepackage{xspace}
\usepackage[all]{xy}

\usepackage{lazystep}

\begin{document}

\conferenceinfo{ICFP 2011}{September 19--21, 2011, Tokyo, Japan.} 
\copyrightyear{2011} 
\copyrightdata{[to be supplied]} 

\titlebanner{DRAFT--- Do not distribute}        % These are ignored unless
% \preprintfooter{ICFP 2011}   % 'preprint' option specified.

\title{Stepping Lazy Programs}
% this \thanks causes an extra blank page to be generated as the first page
%\thanks{This material is based upon work supported by}

%\subtitle{Lazy Stepping subtitle, if any}

\def\authpiece#1#2{\small
\begin{tabular}{c}
#1\\
#2
\end{tabular}}

\def\neu{Northeastern University} 

\authorinfo{
 \authpiece{Stephen Chang}{\neu}
 \authpiece{Eli Barzilay}{\neu}
 \authpiece{John Clements}{California Polytechnic State University}
 \authpiece{Matthias Felleisen}{\neu}}
 {\relax}
 {\small Communicating Author: stchang@ccs.neu.edu}

\maketitle

\begin{abstract}
Debugging lazy functional programs poses serious challenges. In support of the
``stop, examine, and resume'' debugging style of imperative languages, some
debugging tools abandon lazy evaluation. Other debuggers preserve laziness but
present it in a way that may confuse programmers because the focus of
evaluation jumps around in a seemingly random manner.

In this paper, we introduce a supplemental tool, the algebraic program
stepper. An algebraic stepper shows computation as a mathematical
calculation. Algebraic stepping could be particularly useful for novice
programmers or programmers new to lazy programming. Mathematically speaking, an
algebraic stepper renders computation as the standard rewriting sequence of a
lazy $\lambda$-calculus. Our novel lazy semantics introduces lazy evaluation as
a form of parallel program rewriting. It represents a compromise between
Launchbury's store-based semantics and a simple, axiomatic description of lazy
computation as sharing-via-parameters, \`{a} la Ariola et al. Finally, we prove
that the stepper's run-time machinery correctly reconstructs the standard
rewriting sequence.

%% In this paper, to address some of the shortcomings of previous lazy debuggers,
%% we introduce a new semantics for lazy evaluation. Drawing inspiration from the
%% straightforward computation-by-substitution semantics of the
%% $\lambda$-calculus, our novel semantics introduces lazy evaluation as a form of
%% parallel program rewriting. It represents a compromise between Launchburry's
%% store-based semantics and a simple axiomatic description of lazy computation as
%% sharing-via-parameters, a la Ariola et al.

%% Our new semantics is naturally exhibited by an algebraic program stepper, which
%% we propose as a supplemental tool to existing debuggers. Mathematically
%% speaking, an algebraic stepper renders computation as the standard reduction
%% sequence of a lazy lambda calculus. Such a tool could be particularly useful
%% for novice programmers or programmers new to lazy programming. We implement
%% such a tool and prove that the stepper's run-time machinery correctly
%% reconstructs the standard reduction.
\end{abstract}



% \keywords keyword1, keyword2

\section{How Functional Programming Works}

\label{sec:intro}

\ncite{Hughes1989WhyFPMatters} explains why functional programming matters. By
``functional programming'' Hughes specifically means ``lazy'' functional
programming, and by ``matters'' he refers to the distinct advantages of
laziness for programming with reusable components, i.e., functions and
programs. Hughes's examples demonstrate the ease of creating functions and
programs by gluing together existing functions and programs. The key advantage
of laziness is that even if one component appears to produce too much data, the
laziness of the consuming component almost always naturally eliminates the
superfluous data production. Since then, numerous researchers have extended
this argument with examples of their own, usually accompanied by elegant code
in Haskell.

Unfortunately, laziness also increases the distance between the programmer and
the underlying machinery. Specifically, laziness reduces a programmer's ability
to predict when certain expressions are evaluated during program execution. As
long as things work, this cognitive dissonance poses no problems. When a
program exhibits erroneous behavior, however, programmers are often at a
loss. A programmer can turn to a debugger for help, but the evaluation of lazy
programs is often confusing enough that lazy debuggers resort to hiding
laziness from the programmer in order to display useful
information~\cite{Wallace2001NewHat,Ennals2003HsDebug,Allwood2009Needle}.

An ideal debugger should not modify the execution model of a program. The
maintainers of the Glasgow Haskell Compiler~\cite{GHC} share this sentiment,
since the bundled GHCi debugger abides by this ideal. The authors of the GHCi
debugger~\cite{Marlow2007Debugger} state that their debugger ``lets the
programmer see the effects of laziness,'' and therefore, ``shows the programmer
what is actually happening in their program at runtime''. Unfortunately, the
authors also acknowledge that their debugger presents lazy computations in a
way that is difficult to follow, mostly through seemingly random jumps from one
place to another in the program.

To mitigate the drawbacks of existing lazy debuggers, we propose a
supplementary tool, the algebraic stepper. PLT's
DrRacket~\cite{fcffksf:drscheme} comes with such a stepper for the
call-by-value Racket\footnote{DrRacket and Racket were formerly known as
  DrScheme and PLT Scheme, respectively.} language. Given a functional program,
the algebraic stepper displays its standard reduction sequence in the by-value
$\lambda$-calculus~\cite{cff:stepper}. Our experience suggests that this kind
of tool especially benefits novice programmers when they try to understand
small programs, and programmers who are new to the language and are trying to
explore some linguistic feature.

DrRacket's stepper presents evaluation of a program directly as a manipulation
of the source, similar to the calculated manipulations of a student of
mathematics. We conjecture that this model is suitable for addressing the
confusing nature of lazy evaluation. In this paper, we present: (1) our stepper
for Lazy Racket~\cite{bc:lazy-scheme}; (2) its underlying semantics, a novel
lazy $\lambda$-calculus; (3) and a proof that the stepper correctly implements
the standard rewriting semantics of the calculus. Roughly speaking, Lazy Racket
is a call-by-need language that uses the same evaluation mechanism as
Haskell.\footnote{The implementation of Lazy Racket is similar to SRFI 45 of
  the Scheme Standard.} A Lazy Racket module macro-expands its surface syntax
into a plain Racket program enriched with appropriate \delayname and \forcename
constructs~\cite{Hatcliff1997Thunks}. Lazy Racket is mostly used in educational
settings---where we have tested the prototype of the stepper so far---though
some programmers have found it useful to construct parser combinators or game
trees in Lazy Racket modules, and then to export such pieces to plain Racket
modules.

The novel lazy $\lambda$-calculus introduces the idea of selective parallel
reduction to simulate shared reductions. On the one hand, it is nearly trivial
to prove an equivalence between the standard rewriting semantics of our
calculus and a Launchbury-style, store-based
semantics~\cite{Launchbury1993NaturalSemantics}. On the other hand, the
calculus is the appropriate basis for a correctness proof of the stepper. For
the correctness proof, we construct a model of the \delayname-and-\forcename
implementation, further enriched with continuation marks~\cite{cff:stepper},
show that it bisimulates the standard rewriting semantics, and finally, exploit
Clements's strategy for the rest of the proof~\cite{Clements2006Thesis}.

Section~2 briefly introduces Lazy Racket and its stepper. Section~3 presents a
novel lazy $\lambda$-calculus and its essential theorems. Section~4 summarizes
the essential details of the implementation of Lazy Racket and presents a model
of the lazy stepper. Section~5 presents a correctness proof for the stepper and
the penultimate section summarizes related work.

\section{Lazy Racket and Its Stepper} 
\label{sec:examples}
%
%%%%%%%%%%%%%%%%%%%% Example 1 %%%%%%%%%%%%%%%%%%%%

Lazy Racket programs are sequences of \texttt{define}s and expressions that
usually refer to the definitions. Here is a basic example:
\begin{alltt}
        (define (f x) (+ x x))
        (f (+ 1 (+ 2 3)))
\end{alltt}
A programmer invokes the Lazy Racket stepper from the DrRacket IDE. Running the
stepper displays the reduction sequence for the current
program. Figure~\ref{fig:example1} shows a sequence of screenshots stepping
through the above program, with each shot displaying one reduction step.
\begin{figure}[htb]
  \includegraphics[width=3.3in]{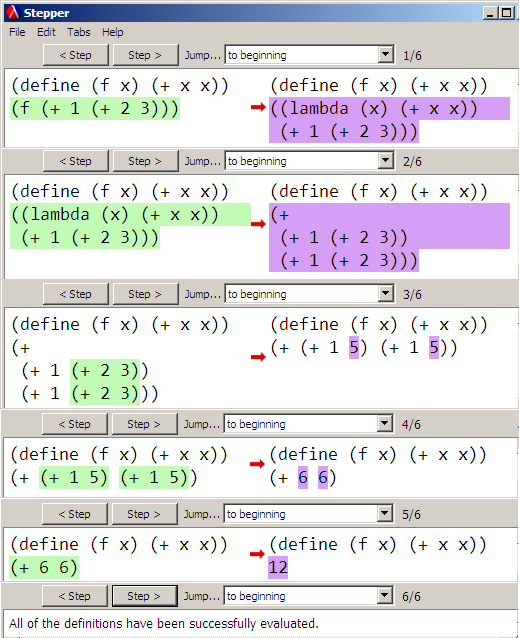}
  \caption{Lazy Stepper Example 1}
  \label{fig:example1}
\end{figure}
A green box highlights the redex(es) on the left-hand side of a reduction step
and a purple box highlights the contractum(s) on the right-hand side. The
programmer can navigate the reduction sequence in either the forward or
backward direction. Additional navigation features are in the planning stages.

In step~2, evaluation of the function argument is delayed so an unevaluated
argument replaces each instance of the variable \texttt{x} in the function
body. In step~3, evaluation of the program at the first \texttt{x} position
requires the value of the argument, so the argument is forced in
steps~3~and~4. In steps~3~and~4, all shared instances of the argument are
reduced simultaneously. That is, \emph{the stepper explains evaluation as an
  algebraic process using a form of parallel reduction}. Since the second
\texttt{x} refers to the same delayed computation as the first \texttt{x}, by
the time evaluation of the program requires a value at the second \texttt{x}
position, a result is already available because the computed value of the first
\texttt{x} was saved. In short, no argument evaluation is repeated.
%
%%%%%%%%%%%%%%%%%%%% Example 2 %%%%%%%%%%%%%%%%%%%%

A second example introduces the lazy \takename function, which extracts the
first \texttt{n} elements of a specified list:
\begin{alltt}
  (define (take! n lst)
    (if (= n 0) 
        null
        (cons (first lst) 
              (take! (- n 1) (rest lst)))))
  (define (f lst) (+ (first lst) (second lst)))
  (f (take! 3 (list 1 2 (/ 1 0) 4)))
\end{alltt}
The reduction sequence for this program, as viewed in the stepper, appears in
figure~\ref{fig:example2}. For space reasons, only interesting steps are
shown.
\begin{figure*}[htb]
  \includegraphics[width=7in]{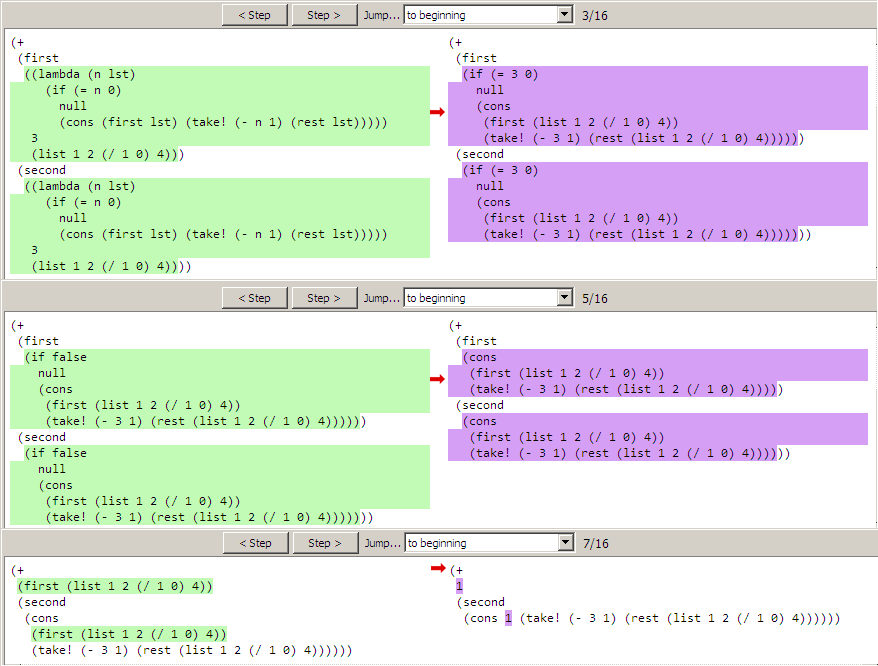}
  \caption{Lazy Stepper Example 2}
  \label{fig:example2}
\end{figure*}

In this example, the result of the \takename computation is the argument to the
function \texttt{f}. The \takename computation extracts the first three
elements of its list argument, but \texttt{f} only uses the first two list
elements, so the third element, which produces an error, should not be
forced. In step~3, the \takename computation is forced because both \texttt{+}
and \carname are strict. In Lazy Racket, \consname behaves lazily and does not
evaluate its arguments~\cite{Friedman1976Cons}, so in step~5, the result of the
\takename computation is a \consname with two thunks: one that retrieves the
first element of the list, and one that contains the next iteration of the
\takename computation. In step~7, the first addition operand is finally reduced
to a value. Notice that the first element in the argument to \texttt{second} is
already reduced as well. The remaining steps force the next iteration of
\texttt{take!} and similarly extract the second element of the list. Since only
the first two elements of the list are needed, no additional \takename
iterations are computed and the division by zero never raises an error.
%
%%%%%%%%%%%%%%%%%%%% Example 3 %%%%%%%%%%%%%%%%%%%%

A third example involves infinite lists:
\begin{alltt}
    (define (add-one x) (+ x 1))
    (define nats (cons 1 (map add-one nats)))
    (+ (second nats) (third nats))
\end{alltt}
More importantly, it involves \texttt{map}, which the stepper has not annotated
because it is a library function. The reduction sequence for this program
appears in figure~\ref{fig:example3}. Again, some function definitions and
reduction steps have been elided from the screenshots.
\begin{figure*}[htb]
  \includegraphics[width=7in]{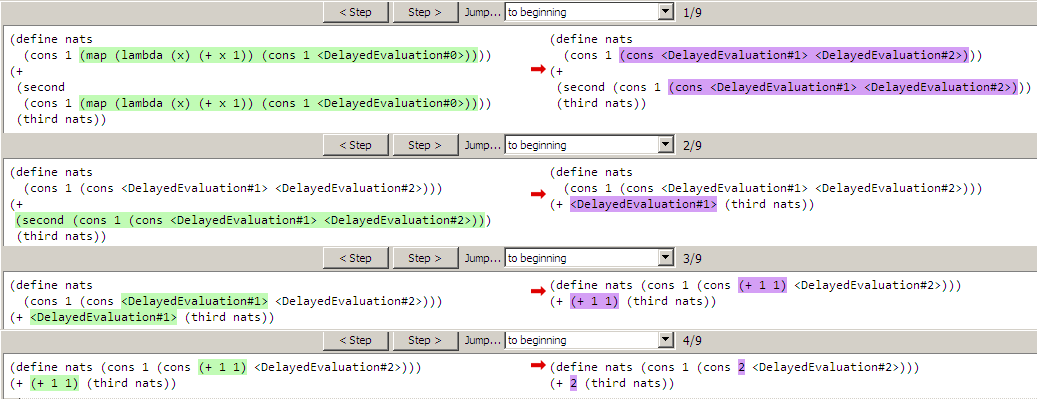}
  \caption{Lazy Stepper Example 3}
  \label{fig:example3}
\end{figure*}
In step~1, the evaluation of \texttt{second} forces the \texttt{map} expression
to a \consname containing two thunks. Unlike the second example, the thunks are
displayed as $\thunk{1}$ and $\thunk{2}$ because their contents are unknown,
i.e., they were not part of the source program. In step~2, the \texttt{second}
expression extracts $\thunk{1}$ from the list, but the thunk is still
unevaluated. In steps~3~and~4, evaluation of the program requires the value of
$\thunk{1}$, so it is forced. Observe how the stepper updates the \texttt{nats}
definition with the result as well. The remaining steps show the similar
evaluation of the other addition operand and are thus omitted.
%
%%%%%%%%%%%%%%%%%%%% GHCi Example %%%%%%%%%%%%%%%%%%%%

As a final example, we use our stepper to understand the behavior of a program
presented by \ncite{Marlow2007Debugger}:
\begin{verbatim}


;; [Listof Char] -> [Listof [Listof Char]]
(define (lines s)
 (cond
   [(null? s) null]
   [else 
    (define-values (l t) 
      (break (lambda (x) (equal? "\n" x)) s))
    (cons l (cond
              [(null? t) null]
              ; drop "\n" char and recur
              [else (lines (cdr t))]))]))

;; [Char -> Boolean] [Listof Char] 
;;  ->* [Listof Char] [Listof Char]
(define (break p? l)
  (let L ([l l] [line null])
    (cond
      [(null? l) (values (reverse line) null)]
      [(p? (car l)) (values (reverse line) l)]
      [else (L (cdr l) (cons (car l) line))])))
\end{verbatim}
The \breakname function consumes two arguments, a predicate on characters and a
string represented as a list of characters, and splits the string at the first
character for which the predicate is true, returning two substrings
simultaneously.\footnote{The \texttt{values} Racket construct enables multiple
  return values.} The delimiting character remains as the first character of
the second substring. The \linesname function uses \breakname to separate an
input string into lines, where a \verb!"\n"! character begins a new
line. Unlike \breakname, the delimiting character is not included in the output
of \linesname.

Evaluating the expression \verb!(lines '("\n" "a"))! produces the expected
result, an empty line (the empty string) and a line with one \verb!"a"!
  character (the \texttt{!!}  function is a recursive \forcename function):
\begin{verbatim}
    > (!! (lines '("\n" "a")))
    '(() ("a"))
\end{verbatim}
However, \verb!(lines '("a" "\n"))! produces only one line:
\begin{verbatim}
    > (!! (lines '("a" "\n")))
    '(("a"))
\end{verbatim}
\ncite{Marlow2007Debugger} show how to use the GHCi debugger to understand this
behavior. Figure~\ref{fig:example4} demonstrates that the stepper provides a
superior vehicle in this situation. It specifically shows how \texttt{break}
returns two values, which when evaluated, produce a one-character string
\verb!"a"! and a one-character string \verb!"\n"!, respectively. From the
remainder of the definition of \linesname, we can deduce that the \verb!"a"!
string is retained, while the \verb!"\n"! is dropped by the subsequent call to
\texttt{cdr}, causing the recursive \linesname call to return an empty list,
i.e. no lines, thus explaining the missing line. If we want \linesname to
output a final, empty line when there is a trailing \verb!"\n"! in the input,
the base case must return a list with an empty line, \verb!'(())!, instead of
just an empty list.
\begin{figure*}[htb]
  \includegraphics[width=7in]{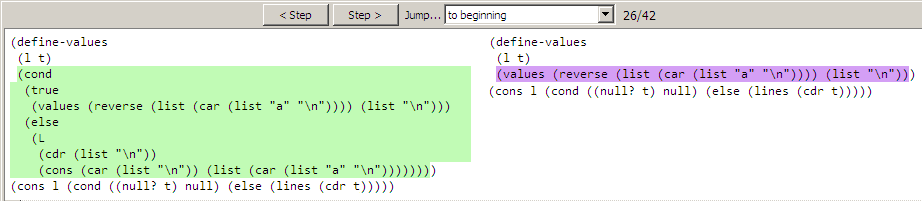}
  \caption{Lazy Stepper Example 4}
  \label{fig:example4}
\end{figure*}
\section{Lazy Racket Semantics}

\label{sec:lambdalr}

Our key theoretical innovation is the novel semantic view of laziness displayed
in our stepper. Following tradition, we present this idea in the context of a
$\lambda$-calculus, \lambdaLR:
\begin{align*}
e =&\;n \mid s \mid b \mid x \mid \lam{x}{e} \mid \app{e}{e} 
        \mid \prim{p^2}{e}{e} \\
 &\mid \cons{e}{e} \mid \nul \mid \app{p^1}{e} \mid \iif{e}{e}{e} \\
n \in&\; \mathbb{Z},\; s \in \textrm{Strings},\;b = \true \mid \false \\
p^2 =&\; + \mid - \mid * \mid / \hspace{0.5cm} p^1 = \carname \mid \cdrname 
\end{align*}
The syntax of \lambdaLR is identical to the core of most functional programming
languages and includes integers, strings, booleans, variables, abstractions,
applications, primitives, lists, and a conditional.

To specify the semantics of \lambdaLR, we first extend $e$ by adding a new
expression:
\begin{align*}
  \eLR &= e \mid \labdef{\eLR} \hspace{0.5cm} \ell \in \textrm{Labels}
\end{align*}
The ``labeled'' expression, \lab{\ell}{\eLR}, consists of a tag $\ell$ and a
subexpression $\eLR$. Labeled expressions are not part of the language syntax
but are necessary for evaluation. Rewriting a labeled expression triggers the
simultaneous rewriting of all other expressions that share the same
label. Otherwise, labeled expressions do not affect program evaluation. The
stepper renders labeled expressions without the label.

We require one constraint on labeled expressions in our language: all
expressions with the same label $\ell$ must be identical. We call this the
consistent labeling property:
\begin{definition}
\label{prop:samelab}
A program is consistently labeled if, for all $\ell_1$, $\ell_2$, $\eLR_1$,
$\eLR_2$, if $\lab{\ell_1}{\eLR_1}$ and $\lab{\ell_2}{\eLR_2}$ are two
subexpressions in a program, and $\ell_1 = \ell_2$, then $\eLR_1 = \eLR_2$.
\end{definition}
%
%
%
% Figure: \lambdaLR Rewriting Rules
%
\begin{figure*}[htp]
  \begin{center}
    \begin{tabular}{c||c|c||c}
\hline
      \multicolumn{3}{l}{\hspace{1.5cm} LHS 
                         \hspace{1.5cm} $\onesteplr$ 
                         \hspace{2.4cm} RHS} & \\
\hline
%& & $\inhole{E}{\:} = \inhole{E_1}{\labdef{(\inhole{E_2}{\:})}}$ & \\
& $\nexists E_1,E_2,\ell$
& $\inhole{E}{\:} = \inhole{E_1}{\labdef{(\inhole{E_2}{\:})}}$ & \\
& s.t. $\inhole{E}{\:} = \inhole{E_1}{\labdef{(\inhole{E_2}{\:})}}$ 
& $\nexists E_3,E_4,\ell'$ 
s.t. $\inhole{E_2}{\:} = \inhole{E_3}{\lab{\ell'}{(\inhole{E_4}{\:})}}$ & \\
%& s.t. $\inhole{E_2}{\:} = \inhole{E_3}{\lab{\ell'}{(\inhole{E_4}{\:})}}$ & \\
& (redex not under label) & (redex occurs under label) & \\
%& & & \\
\hline
%& & & \\
      \makerule{\app{\labvecdef{\lamp{x}{\eLR_1}}}{\eLR_2}}
               {\subst{\eLR_1}{x}{\lab{\ell_1}{\eLR_2}}}
               {\betalr} \\
      & $\fresh{\ell_1}$ & $\fresh{\ell_1}$ & \\[2pt]
%& & & \\
      \makerule{\prim{p^2}{\labvecdef{n_1}}
                         {\labvecdef{n_2}}}
               {\delt{\prim{p^2}{n_1}{n_2}}}
               {\primlr} \\[2pt]
%%      & $p^2 \neq / \vee n_2 \neq 0$ & $p^2 \neq / \vee n_2 \neq 0$ & \\
%& & & \\
      \makerule{\cons{\eLR_1}{\eLR_2}}
               {\cons{\lab{\ell_1}{\eLR_1}}
                     {\lab{\ell_2}{\eLR_2}}}
               {\conslr} \\
      $\eLR_1 \textrm{ unlabeled or } \eLR_2 \textrm{ unlabeled}$
      & $\fresh{\ell_1,\ell_2}$ & $\fresh{\ell_1,\ell_2}$ & \\[2pt]
%& & & \\
      \makerule{\car{\labvecdef{\cons{\eLR_1}{\eLR_2}}}}
               {\eLR_1}
               {\carlr} \\[2pt]
%& & & \\
      \makerule{\cdr{\labvecdef{\cons{\eLR_1}{\eLR_2}}}}
               {\eLR_2}
               {\cdrlr} \\[2pt]
%& & & \\
      \makerule{\iif{\labvecdef{\true}}{\eLR_1}{\eLR_2}}
               {\eLR_1}
               {\iftruelr} \\[2pt]
%& & & \\
      \makerule{\iif{\labvecdef{\false}}{\eLR_1}{\eLR_2}}
               {\eLR_2}
               {\iffalselr} \\[2pt]
\hline
    \end{tabular}
  \end{center}
  \caption{The \lambdaLR Reduction System.}
  \label{fig:lr}
\end{figure*}

\subsection{Rewriting Rules}

To further formulate a semantics, we define the notion of values:
\begin{align*}
v = n \mid s \mid b \mid \lam{x}{\eLR} \mid \nul 
      \mid \cons{\labdef{\eLR}}{\labdef{\eLR}} \mid \lab{\ell}{v}
\end{align*}
Numbers, strings, booleans, abstractions, and \nulname are values. In addition,
\consname expressions where each element is labeled are also values. Finally,
any value tagged with labels is also a value.

In the rewriting of \lambdaLR programs, evaluation contexts are used to
determine which part of the program to rewrite next. Evaluation contexts are
expressions where a hole $\holeE$ replaces one subexpression:
\begin{align*}
E = \holeE &\mid \app{E}{\eLR} \mid 
      \prim{p^2}{E}{\eLR} \mid \prim{p^2}{v}{E} \mid \app{p^1}{E} \\
      &\mid \iif{E}{\eLR}{\eLR} \mid \lab{\ell}{E} 
\end{align*}
The $\app{E}{\eLR}$ context indicates that the operator in an application is
evaluated before it is applied. The $p^1$ and $p^2$ contexts indicate that these
primitives are strict in all argument positions. The \ifname context dictates
strict evaluation of only the test expression. Finally, the \labdef{E} context
dictates that a redex search goes under labeled expressions. Essentially, when
searching for a redex, expressions tagged with a label are treated as if they
were unlabeled.

Evaluation of a \lambdaLR program proceeds according to the program rewriting
system in figure~\ref{fig:lr}. For each possible rewriting step, the program in
the first column is partitioned into a redex and a context, and is rewritten to
either the program in the second or the third column. In the second and third
columns, the redex is always contracted. If the redex does not occur under a
label, then it is the only contracted part of the program (column two). If the
redex occurs under a label, all other instances of the label are similarly
contracted (column three). In the third column, the context is further
subdivided as $\inhole{E}{\:} = \inhole{E_1}{\labdef{(\inhole{E_2}{\:})}}$
where $\ell$ is the label nearest the redex, $E_1$ is the context around the
$\ell$-labeled expression, and $E_2$ is the context under label $\ell$
surrounding the redex. Thus $E_2$ contains no additional labels. An ``update''
function is used to perform the parallel reduction. The update function uses
the notation $\substlab{\eLR_1}{\ell}{\eLR_2}$ to mean that all expressions in
$\eLR_1$ immediately under a label $\ell$ are replaced with $\eLR_2$. The
function is formally defined in figure~\ref{fig:update}. The last clause in the
definition covers all cases not included by the preceding clauses. 
%In the third column, a shorthand ``\_'' is plugged into the hole of $E_1$ because the update function ultimately inserts the desired contractedy subexpression.
%
%
% Figure: ``update'' function
%
\begin{figure}[htb]
  \begin{align*}
    % label, same
    \replacelabeledexpr{\lab{\ell}{\eLR_1}}{\ell}{\eLR_2}    
        &= \lab{\ell}{\eLR_2} \\
    % label, different
    \replacelabeledexpr{\lab{\ell_1}{\eLR_1}}{\ell_2}{\eLR_2}
        &= \lab{\ell_1}{(\replacelabeledexpr{\eLR_1}{\ell_2}{\eLR_2})}, 
        \; \ell_1 \neq \ell_2  \\
    % lambda
    \replacelabeledexpr{\lamp{x}{\eLR_1}}{\ell}{\eLR_2} &=
      \lam{x}{(\replacelabeledexpr{\eLR_1}{\ell}{\eLR_2})} \\
    % app
    \replacelabeledexpr{\app{\eLR_1}{\eLR_2}}{\ell}{\eLR_3} &= 
      (\replacelabeledexpr{\eLR_1}{\ell}{\eLR_3} \\
       &\hspace{5.5mm}\replacelabeledexpr{\eLR_2}{\ell}{\eLR_3}) \\
    % prim2
    \replacelabeledexpr{\prim{p^2}{\eLR_1}{\eLR_2}}{\ell}{\eLR_3} &= 
      (p^2\;\replacelabeledexpr{\eLR_1}{\ell}{\eLR_3} \\
               &\hspace{10mm}\replacelabeledexpr{\eLR_2}{\ell}{\eLR_3}) \\
    % cons
    \replacelabeledexpr{\cons{\eLR_1}{\eLR_2}}{\ell}{\eLR_3} &= 
      (\texttt{cons}\;\replacelabeledexpr{\eLR_1}{\ell}{\eLR_3} \\
           &\hspace{13.5mm}\replacelabeledexpr{\eLR_2}{\ell}{\eLR_3}) \\
    % prim1
    \replacelabeledexpr{\app{p^1}{\eLR_1}}{\ell}{\eLR_2} &=
      \app{p^1}{\replacelabeledexpr{\eLR_1}{\ell}{\eLR_2}} \\
    % if
    \replacelabeledexpr{\iif{\eLR_1}{\eLR_2}{\eLR_3}}{\ell}{\eLR_4} &= 
      (\texttt{if}\;\replacelabeledexpr{\eLR_1}{\ell}{\eLR_4} \\
      &\hspace{10mm}\replacelabeledexpr{\eLR_2}{\ell}{\eLR_4} \\
      &\hspace{10mm}\replacelabeledexpr{\eLR_3}{\ell}{\eLR_4}) \\
    % ow
    \textrm{otherwise, }\replacelabeledexpr{\eLR_1}{\ell}{\eLR_2} &= \eLR_1
  \end{align*}
  \caption{Definition of parallel update function.}
  \label{fig:update}
\end{figure}

The $\betalr$ rule specifies that function application occurs before the
evaluation of arguments. To remember where expressions originate, the argument
receives a label $\ell_1$ before substitution is performed. The notation
$\lab{\vec{\ell}}{\eLR}$ represents an expression wrapped in one or more
labels. During a rewriting step, labels are discarded from values because no
further reduction is possible.

Binary primitive applications are strict in their arguments, as seen in the
\primlr rule. The $\delta$ function interprets binary primitive applications
and is defined in the standard way (division by 0 results in a stuck state).

The \conslr rule shows that, if either argument is unlabeled, both arguments
are wrapped with labels. Adding an extra label around an already labeled
expression will not change the rewriting sequence of the program because the
parallel updating function only uses the innermost label. The \carlr and \cdrlr
rules extract the first and second components from a \consname cell,
respectively, and the \iftruelr and \iffalselr rules similarly choose the first
or second branch of the \ifname expression.

Program rewriting preserves the consistent labeling property.
\begin{lemma}
\label{lem:samelab}
If $\eLR_1$ is consistently labeled and $\eLR_1 \onesteplr \eLR_2$, then
$\eLR_2$ is consistently labeled.
\end{lemma}

The rewriting rules are deterministic because any expression $\eLR$ can be
uniquely partitioned into an evaluation context and a redex. Thus if $\eLR_1$
rewrites to a expression $\eLR_2$, then $\eLR_1$ rewrites to $\eLR_2$ in a
canonical manner.

We can then use $\onesteplr$ to define an evaluator:
$$
\evallr(e)
\begin{cases} 
  v,      & \textrm{if }e \multisteplr v \\
  \bot,   & \textrm{if, for all } e \multisteplr \eLR_1, \eLR_1 \onesteplr \eLR_2  \\
  \error, & \textrm{if } e \multisteplr \eLR_1, \eLR_1 \notin v, \\
          &  \not\exists \eLR_2 \textrm{ such that } \eLR_1 \onesteplr \eLR_2 \\
\end{cases}
$$
where $\multisteplr$ is the reflexive-transitive closure of $\onesteplr$.

%% -----------------------------------------------------------------------------
\subsection{A New Call-by-Need Lambda Calculus}

As is, the $\onesteplr$ relation cannot describe the standard reduction
sequences of any calculus. Each rule in figure~\ref{fig:lr} replaces the {\em
  entire\/} program with another, in a non-compositional manner.  Put
differently, the table of relations does not show how a standard reduction
semantics is created from a basic notion of reduction~\cite{b:lambda}, like
$\beta$, made compatible with the syntactic constructions of a
leftmost-outermost context~\cite{Felleisen2009Semantics,gdp:cbn-cbv}.

In this section, we sketch how our rewriting semantics relates to a plain, yet
novel call-by-need calculus. Put concisely, the calculus replaces the {\it
  deref\/} notion of reduction in \citepos{Ariola1997CBNeedJFP} call-by-need
calculus\footnote{The calculus of \ncite{Maraist1998CBNeedJFP} is unrelated in
  this case.} with a rule that exploits the function parameter for sharing but
implements evaluation via substitution.  The calculus comes with a standard
reduction theorem, and it is possible to show that our rewriting semantics
essentially relates the program to its answer via the same steps as the
standard reduction sequences of the calculus. Due to a lack of space, we merely
spell out the basic ideas without stating any theorems or proofs.

For simplicity, we take the syntax of the $\lambda$-calculus as the
 starting point: 
$$\eAF = x \mid \lam{x}{\eAF} \mid \app{\eAF}{\eAF}$$
 Evaluation of a \lambdaAF program terminates when it is reduced to an answer
 $\AAF$:
\begin{align*}
\VAF &= \lam{x}{\eAF} \\
\AAF &= \VAF \mid \app{\lamp{x}{\AAF}}{\eAF}
\end{align*}
Programs reduce to answers instead of values because reduction does not remove
 application terms. The specification of a notion of reduction relies on
 the notion of an evaluation context $\EAF$:
$$
 \EAF = \holeE \mid \app{\EAF}{\eAF} 
               \mid \app{\lamp{x}{\inhole{\EAF}{x}}}{\EAF} 
               \mid \app{\lamp{x}{\EAF}}{\eAF}
$$
 The evaluation contexts, especially the third one, specify that arguments
 are not evaluated until they are needed by some variable in the function
 body. 

 Here are the three notions of reductions (axioms) from Ariola and Felleisen's
 call-by-need calculus:
%
%$$
% \begin{array}{lrcl}
\begin{align*}
% (\textit{deref})\\ 
\tag{\textit{deref}} \\[-3pt]
% & 
 \app{\lamp{x}{\inhole{\EAF}{x}}}{\VAF} 
 &=
 \app{\lamp{x}{\inhole{\EAF}{\VAF}}}{\VAF} \\
%% ------------------------------------------------------------------
% (\textit{lift}) \\
\tag{\textit{lift}} \\[-3pt]
% & 
 \app{\app{\lamp{x}{\AAF}}{\eAF}}{{\eAF}'}
 &=
 \app{\lamp{x}{\app{\AAF}{{\eAF}'}}}{\eAF} \\
%% ------------------------------------------------------------------
% (\textit{assoc})
\tag{\textit{assoc}} \\[-3pt]
% & \\
% \multicolumn{4}{c}{
  \app{\lamp{x}{\inhole{\EAF}{x}}}{\app{\lamp{y}{\AAF}}{\eAF}}
   &= 
   \app{\lamp{y}{\app{\lamp{x}{\inhole{\EAF}{x}}}{\AAF}}}{\eAF}
\end{align*}
% \end{array}
%$$
 The {\it deref\/} axiom substitutes the evaluated argument for the variable in
 the function body. The other two axioms deal with answers that may appear
 on the left-hand or right-hand side of an application.

One problem is that the {\it deref\/} axiom leaves the application
 alone, even after the argument has been reduced to a value. Clearly, doing
 so contradicts both the natural implementations (which use a mix of graph
 rewriting and memoizing) and the widely used Launchbury semantics (which
 uses a store-based semantics to mimic memoizing). To get closer to this
 semantics, we propose to replace the \textit{deref} axiom with the
 following $\betaneed$ axiom:
\begin{align*}
  \tag{$\betaneed$}
     \app{\lamp{x}{\inhole{\EAF}{x}}}{\VAF} 
  &= \subst{\inhole{\EAF}{x}}{x}{\VAF}
\end{align*}
 This axiom says that when a parameter occurs in a ``demand''
 position---the hole of an evaluation context---and the argument is a
 value, then plain old substitution captures the essence of parameter
 passing in a call-by-need language. 

The new axiom is indeed a proper notion of reduction and is applicable in any
context. Like Ariola and Felleisen's calculus, our revised lazy
$\lambda$-calculus is confluent (satisfies the Church-Rosser property) and
comes with a Curry-Feys style standard reduction theorem. Furthermore, there is
a simple bisimulation that relates standard reduction sequences to the
semantics of figure~\ref{fig:lr}.

%
%
%
%
%
%
%
% Figure: stepper architecture (for next section)
%
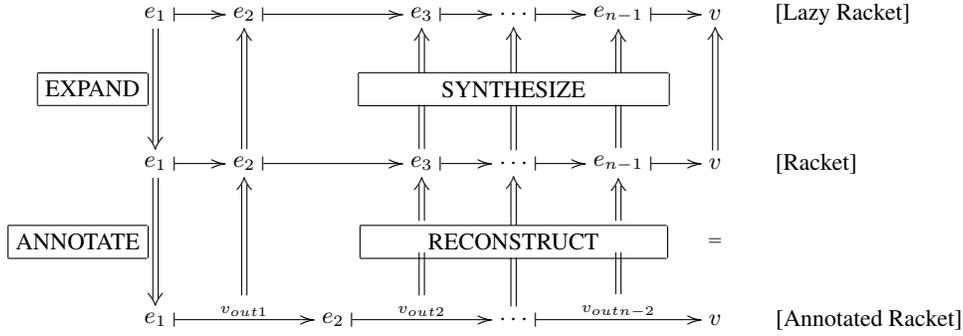
\begin{figure*}[htb]
  $$\xymatrixcolsep{5mm}\xymatrixrowsep{15mm}\xymatrix{
%  $$\xymatrix{
    e_1      \ar@{|->}[r] 
             \ar@{=>}[d]_*+[F]{\textrm{EXPAND}} 
    & e_2    \ar@{|->}[rr]
    &
    & e_3    \ar@{|->}[r]
    & \cdots \ar@{|->}[r] 
    & e_{n-1} \ar@{|->}[r]
    & v
    & \textrm{[Lazy Racket]}\hspace{7mm} \\
    %%%%%%%%%%
    e_1      \ar@{|->}[r] 
             \ar@{=>}[d]_*+[F]{\textrm{ANNOTATE}} 
    & e_2    \ar@{|->}[rr] 
             \ar@{=>}[u]
    &
    & e_3    \ar@{|->}[r] 
             \ar@{=>}[u]|-*+<12pt>{}
    & \cdots \ar@{|->}[r] 
             \ar@{=>}[u]|-*+[F]{\hspace{10mm}\textrm{SYNTHESIZE}\hspace{10mm}}
    & e_{n-1} \ar@{|->}[r] 
             \ar@{=>}[u]|-*+<12pt>{}
    & v      \ar@{=>}[u] 
    & \textrm{[Racket]}\hspace{14mm} \\
    %%%%%%%%%%
    e_1      \ar@{|->}[rr]
    &  \stackrel{v_{out1}}{\stackrel{}{{}^{}}} 
             \ar@{=>}[u]|-<{}%v_{out1}}
    & e_2    \ar@{|->}[rr]
    &  \stackrel{v_{out2}}{\stackrel{}{{}^{}}} 
             \ar@{=>}[u]|-*+<12pt>{}
    & \cdots \ar@{|->}[rr]
             \ar@{=>}[u]|-*+[F]{\hspace{8mm}\textrm{RECONSTRUCT}\hspace{8mm}}
    &  \stackrel{v_{outn-2}}{\stackrel{}{{}^{}}} 
             \ar@{=>}[u]|-*+<12pt>{}
    & v      \ar@{}[u]|-{=}
    & \textrm{[Annotated Racket]}
  }$$
  \caption{Stepper Implementation Architecture}
  \label{fig:stepperarch}
\end{figure*}

\section{Lazy Stepper Implementation}
\label{sec:stepperimpl}

Figure~\ref{fig:stepperarch} summarizes the software architecture of our
stepper. The first row depicts a \lambdaLR Lazy Racket rewriting sequence. To
construct this rewriting sequence, the lazy stepper first macro-expands a Lazy
Racket program into a functional Racket program, enriched with \delayname and
\forcename, as mentioned in section~\ref{sec:intro}. In turn, the stepper for
(eager) Racket annotates the expanded program. Executing an annotated Racket
program emits a series of output values, from which the reduction sequence for
the unannotated Racket program is reconstructed. Once the lazy stepper has the
plain Racket reduction sequence, it synthesizes each step to assemble the
desired Lazy Racket rewriting sequence.

The correctness of the lazy stepper thus depends on two claims:

\begin{enumerate}
\item The reduction sequence of a plain Racket program can be reconstructed
  from the output produced when evaluating an annotated version of that
  program.
\item The rewriting sequence of a Lazy Racket program is equivalent to the
  reduction sequence of the corresponding plain Racket program, modulo
  macro-expansion and synthesis steps.
\end{enumerate}
The first point corresponds to the work of \ncite{cff:stepper} and is depicted
by the bottom half of figure~\ref{fig:stepperarch}. The second point is
depicted by the top half of figure~\ref{fig:stepperarch}. The rest of the
section formally presents the architecture in enough detail so that our stepper
can be implemented for other programming languages, and so that we can prove
its correctness. The actual correctness theorem and proof can be found in the
next section.
%
%
% Figure: CS machine transitions
%
\begin{figure*}[htb]
  \begin{center}
  \begin{tabular}{cccl}
    & $\csstep$ & & \\
    \hline \\
    %
    % \beta_v
    $\csdefE{ \app{\lamp{x}{\cDF}}{\vDF} }$ & $\csstep$ &
    $\csdefE{ \subst{\cDF}{x}{\vDF} }$ &
    $\beta_v$ \\
    %
    % \delta
    $\csdefE{ \prim{p^2}{\vDF_1}{\vDF_2} }$ & $\csstep$ &
    $\csdefE{ \delt{\prim{p^2}{\vDF_1}{\vDF_2}} }$ &
    \primcs \\
    %
    % first
    $\csdefE{ \car{\cons{\vDF_1}{\vDF_2}} }$ & $\csstep$ &
    $\csdefE{ \vDF_1 }$ &
    \carcs \\
    %
    % rest
    $\csdefE{ \cdr{\cons{\vDF_1}{\vDF_2}} }$ & $\csstep$ &
    $\csdefE{ \vDF_2 }$ &
    \cdrcs \\
    %
    % if-true
    $\csdefE{ \iif{\true}{\cDF_1}{\cDF_2} }$ & $\csstep$ &
    $\csdefE{ \cDF_1 }$ &
    \iftruecs \\
    %
    % if-false
    $\csdefE{ \iif{\false}{\cDF_1}{\cDF_2} }$ & $\csstep$ &
    $\csdefE{ \cDF_2 }$ &
    \iffalsecs \\
    %
    % delay
    $\csdefE{ \delay{\cDF} }$ & $\csstep$ &
    $\cs{ \inhole{\EDF}{\delaylocdef} }{\storeadddef{\cDF}}$ &
    \delaycs \\
%    $\cDF \neq \delayname$ & \fresh{\ell} & \\
    & & $\ell \notin \dom{\store}$ & \\
    %% %
    %% % delay-nested
    %% $\csdefE{ \delay{\cDF} }$ &
    %% $\csdefE{ \cDF }$ & 
    %% \delaynestedcs \\
    %% $\cDF = \delayname$ & & \\
    %
    % force-delay
    $\csdefE{ \force{\delaylocdef} }$ & $\csstep$ &
    $\csdefE{ \force{\forceEDFdef{\storelookupdef}} }$ &
    \forcedelaycs \\
    %
    % force-update
    $\csdefE{ \forceEDFdef{\vDF} }$ & $\csstep$ &
    $\cs{ \inholedef{\vDF} }{\storeupdatedef{\vDF}}$ &
    \forceupdatecs \\
    %
    % force-nondelay
    $\csdefE{ \force{\vDF} }$ & $\csstep$ &
    $\csdefE{ \vDF }$ &
    \forcenondelaycs \\
    $\vDF \notin \textrm{Locations}$ & &\\
  \end{tabular}
  \end{center}
  \caption{CS Machine Transitions}
  \label{fig:cstransitions}
\end{figure*}

\subsection{Racket + \delayname/\forcename}

When the stepper is invoked on a Lazy Racket program, the source is first
macro-expanded to a plain Racket program. Racket programs are eagerly
evaluated, so lazy evaluation is simulated with the insertion of \delayname and
\forcename constructs. We model this expanded language with \lambdaDF, a core
calculus of functional Racket with \delayname and \forcename:
\begin{align*}
\eDF = 
\; n &\mid s \mid b \mid x \mid \lam{x}{\eDF} \mid \app{\eDF}{\eDF} 
      \mid \prim{p^2}{\eDF}{\eDF} \\
     &\mid \cons{\eDF}{\eDF} \mid \nul \mid \app{p^1}{\eDF} \\
     & \mid \iif{\eDF}{\eDF}{\eDF} 
     \mid \delay{\eDF} \mid \force{\eDF} \\
n \in&\; \mathbb{Z},\; s \in \textrm{Strings},\;b = \true \mid \false \\
p^2 = \;&+ \mid - \mid * \mid / \hspace{0.5cm} p^1 = \carname \mid \cdrname
\end{align*}
The syntax of \lambdaDF is similar to \lambdaLR except that \delayname and
\forcename replace labeled expressions. The \delayname construct suspends
evaluation of its argument in a thunk; applying \forcename to a thunk evaluates
the suspended computation and memoizes the result. In addition, applying
\forcename to a suspended computation wrapped in multiple, nested thunks forces
all the thunks, while applying \forcename to a value returns that value.

%%%%% \subsection{CS Machine} %%%%%

The semantics of \lambdaDF combines the usual call-by-value world with store
effects. We describe it with a high-level abstract machine, specifically, a CS
machine~\cite{Felleisen2009Semantics}. The C in the CS machine stands for
control string and the S is a store that represents physical memory. In our
machine the control string is an expression that may contain locations, i.e.,
references to delayed expressions in the store. In contrast to the standard CS
machine, the store in our machine only holds delayed computations.

Here is the specification of our CS machine:
\begin{align*}
\tag{Machine States}
\confDF &= \cs{\CDF}{\store} \\
\tag{Transition Sequences}
\seqDF &= \seq{\confDF_1}{\confDF_n} \\
\tag{Control Strings}
\CDF &= \inhole{\EDF}{\cDF} \\
\tag{Machine Expressions}
\cDF    &= \eDF \mid \locdef \\
\tag{Stores}
\store  &= \ls{\pair{\ell}{\cDF}} \\
\ell    &\in \textrm{Locations} \\
\tag{Evaluation Contexts}
\EDF &=\holeE \mid \app{\EDF}{\cDF} \mid \app{\vDF}{\cDF} \\
     &\;\;\;\;\;\;\;\:\mid \prim{p^2}{\EDF}{\cDF} \mid \prim{p^2}{\vDF}{\EDF} \\
     &\;\;\;\;\;\;\;\:\mid \cons{\EDF}{\cDF} \mid \cons{\vDF}{\EDF} 
                      \mid \app{p^1}{\EDF} \\
     &\;\;\;\;\;\;\;\:\mid \iif{\EDF}{\cDF}{\cDF} \\
     &\;\;\;\;\;\;\;\:\mid \force{\EDF} \mid \forceEDFdef{\EDF} \\
\tag{Values}
\vDF &= n \mid s \mid b \mid \lam{x}{\cDF} 
          \mid \cons{\vDF}{\vDF} \mid \nul \mid \delaylocdef
\end{align*}

The store in a machine configuration is represented as a list of pairs. In the
above specification, ellipses means ``zero or more of the preceding
element''. The evaluation contexts are the standard call-by-value contexts,
plus two \forcename contexts. The first \forcename context resembles the
\forcename expression in a program and indicates that some arbitrary expression
is being forced. For the evaluation of a specific delayed computation, the
second \forcename context is used. It remembers a location so the store can be
updated after the evaluation is complete. Evaluation under a
$\forceEDFdef{\holeE}$ context corresponds to evaluation under a label in
\lambdaLR.

The starting machine configuration for a Racket program $\eDF$ is
$\csstart{\eDF}$ where the program, in an empty context, is set as the initial
control string, and the store is initially empty. Evaluation stops when the
control string is a value. Values are numbers, strings, booleans, abstractions,
lists, or store locations. Our CS machine transitions are in
figure~\ref{fig:cstransitions}. Every program $\eDF$ has a deterministic
transition sequence because the left hand sides of all the transition rules are
mutually exclusive and cover all possible control strings in the C register.

The $\beta_v$, \primcs, \carcs, \cdrcs, \iftruecs, and \iffalsecs transitions
are standard call-by-value transitions. The \delaycs transition reduces a
\delayname expression to an unused location $\locdef$. The delayed computation
is saved at that location in the store. When the argument to a \forcename
expression is a location, the suspended expression at that location is
retrieved from the store and plugged into a special \forcename evaluation
context, as specified by the \forcedelaycs transition. The outer \forcename
evaluation context is retained in case there are nested \delayname{}s. The
special context saves the store location of the forced expression, so the store
can be updated with the resulting value, as dictated by the \forceupdatecs
transition. Finally, the \forcenondelaycs transition specifies that forcing a
non-location value results in the removal of the outer \forcename context.
%
%
% Figure: CSKM Machine transitions
%
\begin{figure*}[htb]
  \begin{center}
  \begin{tabular}{cccl}
    & $\cskmstep$ & &\\
    \hline \\
    %
    % wcm-exp
    $\cskmexpLHSdef{\wcm{\cDFCM_1}{\cDFCM_2}}$ &
    $\cskmstepdef$ &
    $\cskmexpRHSdef{\cDFCM_1}{\wcmKDFCMdefone{\cDFCM_2}}$ &
    \wcmexpcskm \\
    %
    % wcm-val
    $\cskmvalLHSdef{\vDFCM}{\wcmKDFCMdef}$ &
    $\cskmstepdef$ &
    $\cskm{\cDFCM}{\store}{\KDFCM}{\vDFCM}$ &
    \wcmvalcskm \\
    %
    % ccm
    $\cskmexpLHSdef{\ccm}$ &
    $\cskmstepdef$ &
    $\cskmexpRHSdef{\vDFCM}{\KDFCM}$ &
    \ccmcskm \\
    & & \where{\vDFCM}{\pi(\KDFCM,\contmark)} & \\
    %
    % out-exp
    $\cskmexpLHSdef{\out{\cDFCM}}$ &
    $\cskmstepdef$ &
    $\cskmexpRHSdef{\cDFCM}{\outKDFCMdef}$ &
    \outexpcskm \\
    %
    % out-val
    $\cskmvalLHSdef{\vDFCM}{\outKDFCMdef}$ &
    $\cskmstepout{\vDFCM}$ &
    $\cskmvalRHSdef{42}$ &
    \outvalcskm \\
    %
    % loc-exp
    $\cskmexpLHSdef{\lochuh{\cDFCM}}$ &
    $\cskmstepdef$ &
    $\cskmexpRHSdef{\cDFCM}{\lochuhKDFCMdef}$ &
    \lochuhexpcskm \\
    %
    % loc-val-true
    $\cskmvalLHSdef{\ell}{\lochuhKDFCMdef}$ &
    $\cskmstepdef$ &
    $\cskmvalRHSdef{\true}$ &
    \lochuhvaltruecskm \\
    %
    % loc-val-false
    $\cskmvalLHSdef{\vDFCM}{\lochuhKDFCMdef}$ &
    $\cskmstepdef$ &
    $\cskmvalRHSdef{\false}$ &
    \lochuhvaltruecskm \\
    $\vDFCM \notin \textrm{Locations}$
  \end{tabular}
  \end{center}
  \caption{CSKM Machine Transitions}
  \label{fig:cskmtransitions}
\end{figure*}

\subsection{Continuation Marks}

A stepper for a functional language needs access to the control stack of its
evaluator in order to reconstruct the evaluation steps. In a low-level stepper
implementation, the stepper would be granted complete, privileged access to the
control stack. As \ncite{cff:stepper} argued, however, such privileged
access is unnecessary and often undesirable.

Continuation marks are a lightweight stack-access mechanism. The stepper for
Lazy Racket reuses Clements's stepper for Racket, which utilizes continuation
marks to reconstruct a program's control stack, i.e., the evaluation
context. There are two available operations for these novel values:
\begin{enumerate}
	\item \cmstore a value in the current frame of the control stack,
	\item \cmretrieve all stored continuation marks.
\end{enumerate}
Using these two operations it is possible to implement a stepper without
coupling it directly to the evaluator.

\ncite{cff:stepper} present such a stepper for the eager Racket evaluator. The
stepper first annotates a source program with continuation mark \cmstore and
\cmretrieve operations at appropriate points. Then, at each \cmretrieve point,
the stepper reconstructs and outputs a reduction step from information stored
in the continuation marks. Our stepper extends Clements's model with \delayname
and \forcename constructs. The annotation and reconstruction functions are
formally defined in section~\ref{sec:correctness}.

\subsection{Racket + \delayname/\forcename + Continuation Marks}

After a Lazy Racket program is expanded to a plain Racket program, the lazy
stepper annotates the plain Racket program with continuation mark
operations. The language \lambdaDFCM extends \lambdaDF in a stratified manner
and models the language for annotated programs.
\begin{align*}
\eDFCM &= \eDF \mid \ccm \mid \wcm{\eDFCM}{\eDFCM} \mid \out{\eDFCM} 
               \mid \lochuh{\eDFCM}
\end{align*}
\lambdaDFCM adds four additional kinds of expressions to \lambdaDF: \wcmname,
\ccmname, \outname, and a \lochuhname predicate. When a \wcmname, or ``with
continuation mark'', expression is evaluated, its first argument is evaluated
and stored in the current stack frame before its second argument is
evaluated. A \ccmname, or ``current continuation marks'', expression evaluates
to a list of all continuation marks currently stored in the stack. When
reducing an \outname expression, its argument is evaluated and sent to an
output channel. An \outname expression is evaluated only for this side effect,
so the result of its evaluation is thus inconsequential. The \lochuhname
predicate identifies locations and is needed by annotated programs.

\subsection{CSKM Machine}

To model continuation marks, having an explicit control stack is helpful, so we
convert our CS machine to a CSK machine, where the evaluation context is
separated and removed from the control string in the C register and relocated
to a new K register. The conversion to a CSK machine is straightfoward and is
accomplished using known techniques~\cite{Felleisen2009Semantics}. In addition,
we pair each context with a continuation mark $\contmark$, which is stored in a
fourth ``M'' register, giving us a CSKM machine.

For the control stack in the K register, we use an ``inverted'' evaluation
context structure, meaning that the innermost context is now most easily
accessible, giving us a more realistic stack structure. This new representation
is called a continuation and there is a one-to-one correspondence between
evaluation contexts and continuations. For example, an evaluation context
$\inhole{\EDFCM}{\app{\holeE}{\cDFCM}}$ becomes
$\apponeKDFCM{\cDFCM}{\KDFCM}{\contmark}$, where the $\EDFCM$ context
corresponds to the $\KDFCM$ continuation. The other evaluation contexts are
similarly converted. Note the extra continuation mark $\contmark$ associated
with the continuation. Here are all the continuations $\KDFCM$:
\begin{align*}
\KDFCM = \\
  \mtKDFCM &\mid \apponeKDFCMdef \mid \apptwoKDFCMdef \\
  &\mid \primtwooneKDFCMdef \mid \primtwotwoKDFCMdef \\
  &\mid \consoneKDFCMdef \mid \constwoKDFCMdef  \\
  &\mid \primoneKDFCMdef \mid \ifKDFCMdef \\
  &\mid \forceKDFCMdef   \mid \forcetwoKDFCMdef \\
  &\mid \wcmKDFCMdef     \mid \outKDFCMdef \mid \lochuhKDFCMdef
\end{align*}

The configurations of the CSKM machine are:
\begin{align*}
\tag{Machine States}
\confDFCM &= \cskm{\cDFCM}{\store}{\KDFCM}{\contmark} \\
\tag{Transition Sequences}
\seqDFCM &= \seq{\confDFCM_1}{\confDFCM_n} \\
\tag{Control Strings}
\cDFCM &= \eDFCM \mid \locdef \\
\tag{Values}
\vDFCM &= \vDF \\
\tag{Stores}
\store      &= \ls{\pair{\ell}{\cDFCM}} \\
\tag{Mark Register}
\contmark   &= \nomark \mid \vDFCM
\end{align*}

Control strings are again extended to include location expressions, values are
the same as CS machine values, and stores map locations to control string
expressions. The mark register $\contmark$ is either empty or contains a
value. For simplicity, we assume that only one mark can be associated with a
continuation frame.

The transitions for our CSKM machine are in
figure~\ref{fig:cskmtransitions}. In order to formally model output, we add an
extra tag to each transition, so our machine operates as a labeled transition
system~\cite{Keller1976FormalVerification}. When the machine emits output, the
transition is tagged with the outputted value; otherwise, the transition tag is
$\emptyset$. If a transition has no output tag, it means the output is
inconsequential in the current context. In our machine, only an \outname
expression emits output. For space reasons, we only include the transitions for
the new constructs: \wcmname, \ccmname, \outname, and \lochuhname. The other
transitions are easily derived from the transitions for the CS
machine~\cite{Felleisen2009Semantics}.

The starting configuration for a program $\eDFCM$ is $\cskmstart{\eDFCM}$ and
evaluation halts when the control string is a value and the control stack is
$\mtKDFCM$. Every program $\eDFCM$ has a deterministic transition sequence
because the left hand sides of all the transition rules are mutually exclusive
and cover all possible C and K register combinations.

The \wcmexpcskm transition sets the first argument as the control string and
saves the second argument in a \wcmname continuation. In the resulting machine
configuration, the $\contmark$ register is reinitialized to $\nomark$ because a
new frame is pushed onto the stack. When evaluation of the first \wcmname
argument is complete, the resulting value is set as the new continuation mark,
as dictated by the \wcmvalcskm transition. This new continuation mark
overwrites any previous mark.

The \ccmcskm transition uses the $\pi$ function to retrieve all continuation
marks in the stack. The $\pi$ function is defined as follows:
%
%\begin{figure}[htb]
  \begin{align*}
    \pifn{\mtKDFCM}{\contmark} &= \cons{\contmark}{\nul} \\
    \pifndef{\apponeKDFCMdef}{\contmark'}  \\
    \pifndef{\apptwoKDFCMdef}{\contmark'}  \\
    \ldots \\
    %% \pifndef{\primtwooneKDFCMdef}{\contmark'} \\
    %% \pifndef{\primtwotwoKDFCMdef}{\contmark'} \\
    %% \pifndef{\consoneKDFCMdef}{\contmark'} \\
    %% \pifndef{\constwoKDFCMdef}{\contmark'} \\
    %% \pifndef{\primoneKDFCMdef}{\contmark'} \\
    %% \pifndef{\wcmKDFCMdef}{\contmark'}     \\
    %% \pifndef{\outKDFCMdef}{\contmark'}     \\
    %% \pifndef{\forceKDFCMdef}{\contmark'}   \\
    %% \pifndef{\forcetwoKDFCMdef}{\contmark'} \\
    %% \pifndef{\ifKDFCMdef}{\contmark'} \\
    \pifndef{\lochuhKDFCMdef}{\contmark'}
  \end{align*}
%  \caption{The $\pi$ function retrieves all continuation marks from the control
%    stack.}
%  \label{fig:pi}
%\end{figure}
%
Only the first few cases are shown. The rest of the definition, for other
continuations, is similarly defined. The \outexpcskm transition sets the
argument in an \outname expression as the control string and pushes a new
\outname continuation frame onto the control stack. Again, the continuation
mark register is initialized to $\nomark$ due to the new stack frame. When the
argument is evaluated, the resulting value is emitted as output, as modeled by
the label on the \outvalcskm transition. Finally, the \lochuhexpcskm,
\lochuhvaltruecskm, and \lochuhvalfalsecskm transitions are defined in the
expected manner, producing $\true$ if the \lochuhname predicate is applied to a
location, and $\false$ otherwise.

\section{Correctness}
\label{sec:correctness}

Unlike most IDE tools, an algebraic stepper comes with a concise specification:
the \lambdaLR rewriting system. Specifically, the stepper displays \lambdaLR
rewriting sequences after removing all labels from the terms. It is therefore
relatively straightforward to state a correctness theorem for the stepper,
assuming a function $\unlabelname$ that strips an \lambdaLR program of its
labels.

\begin{theorem}[Stepper Correctness]          \label{theorem:correctness}
 If the stepper displays the sequence $e, \unlabel{\eLR_1}, \ldots,
 \unlabel{\eLR_n}$ for some Lazy Racket program $e$, then $e \multisteplr
 \eLR_1 \multisteplr \cdots \multisteplr \eLR_n$.
\end{theorem}

The statement of the theorem's conclusion involves multistep rewriting
 because some rewriting steps, such as \conslr, merely add labels and
 change nothing else about the term. 

The proof of the theorem consists of two distinct steps. First, we show that
the output of a macro-expanded, annotated Lazy Racket program uniquely
describes the execution of a macro-expanded Lazy Racket program without
annotations. That is, we retrieve a machine reduction sequence at the level of
Racket with \delayname and \forcename. Second, we prove that this reduction
sequence is equivalent to the rewriting sequence of the original Lazy Racket
program, modulo label assignment. The following two subsections spell out the
two lemmas and present proof sketches. The proof of
theorem~\ref{theorem:correctness} combines the two main lemmas from these
subsections in a straightforward fashion.

%% -----------------------------------------------------------------------------
\subsection{Annotation and Reconstruction Correctness} \label{subsec:annorecon}

 To state the correctness lemma for the CSKM machine, we need two
 functions. First $\tracename : \seqDFCM \rightarrow
 \vDFCM_1,\ldots,\vDFCM_n$ consumes a series of CSKM steps and produces the
 trace of output values: 

\begin{align*}
\trace{\cdots \cskmstep \confDFCM_i \cskmstepout{\vDFCM} \confDF_{i+1} \cskmstep \cdots} \\ 
= \ldots, \vDFCM, \ldots
\end{align*}
 Second, $\annoname : \eDF \rightarrow \eDFCM$ annotates a plain Racket 
 program and $\reconname : \vDFCM_1,\ldots,\vDFCM_n \rightarrow \seqDF$
 reconstructs a CS machine transition sequence for a Racket program. 

\begin{lemma}[Annotation/Reconstruction Correctness] \label{lemma:annorecon}
 For any Racket program $\eDF$, if $\csstart{\eDF} \csstep \cdots \csstep
 \confDF$, then 
 $$\begin{array}{rcl}
  \recon{\trace{
          \cskmstart{\anno{\eDF}} \cskmstep \cdots \cskmstep \confDFCM}} &&\\
        \multicolumn{3}{c}{=  \csstart{\eDF} \csstep \cdots \csstep \confDF}
   \end{array}
 $$
\end{lemma}

Our annotation and reconstruction functions extend the functions of
 \ncite{cff:stepper}. Figures~\ref{fig:anno}~and~\ref{fig:recon} summarize
 these additions. We omit the parts defined by Clements et al.\ and instead
 review the functions with some illustrative examples.

\begin{figure}[t]
$\annoname : \eDF \rightarrow \eDFCM$
\begin{alltt}
\(\anno{\force{\eDF}} =\)
  \((\)let*
    \((\)[\(v0\) \(\wcm{\lst{\texttt{"force"}}}{\anno{\eDF}}\)]
     [\(v1\)
      \((\)if \((not (\)loc? \(v0))\)
          \(v0\)
          \((\)wcm 
            \(\boxed{\lst{\texttt{"force"}}}\sb{1}\)
            \((\)wcm 
              \(\boxed{\lst{\texttt{"force"} v0}}\sb{2}\)
              \((\)let* 
                \((\)[\(v2\) \(\force{v0}\)] ; v0 is location
                 [\(t0\) 
                  \((\)\(\boxed{\texttt{output}}\sb{3}\)
                    \((\consname \lst{\texttt{"val"} \boxed{v0}\sb{4} \quo{v2}}\)
                         \(\boxed{\cdr{\ccm}}\sb{5}))\)]\()\)
                \(v2)))\))] 
     [\(t1\) \(\out{\cons{\quo{v1}}{\ccm}}\)]\()\)
    \(v1)\)

\(\anno{\delay{\eDF}}\) =
  \((\)let*
    \((\)[\(t0\) \(\out{\cons{\quo{\delay{\eDF}}}{\ccm}}\)]
     [\(\ell\) \(\alloc\)]
     [\(t1\) \((\)output
           \(\cons{\lst{\texttt{"loc"} \ell \boxed{\quo{\eDF}}\sb{6}}}{\ccm})\)]\()\)
    \(\delay{\anno{\eDF}})\)
\end{alltt}
  \caption{Annotation function for \delayname and \forcename.}
  \label{fig:anno}
\end{figure}

\begin{figure}[htb]
\begin{align*}
%
% recon
%
\reconname &: \vDFCM_1,\ldots,\vDFCM_n \rightarrow \seqDF \\
\recon{\ldots,\vDFCM_i,\ldots} = 
  \ldots,
  \langle & \inhole{\reconE{\cdr{\vDFCM_i}}}{\reconC{\car{\vDFCM_i}}},\\
          & \reconS{\vDFCM_1,\ldots,\vDFCM_i} \rangle,\\ 
  \ldots &
\end{align*}
%
% reconE
%
\begin{align*}
\reconEname : \vDFCM &\rightarrow \EDF \\
\reconE{\cons{\lst{\forcestr}}{\vDFCM}} &= \force{\reconE{\vDFCM}} \\
\reconE{\cons{\lst{\forcestr\;\ell}}{\vDFCM}} &= \forceEDFdef{\reconE{\vDFCM}}
\end{align*}
%
% reconC
%
\begin{align*}
\reconCname : \vDFCM &\rightarrow \cDF \\
\reconC{\lst{\textrm{``val''} \; \ell \; \vDF}} &= \unquo{\vDF} \\
\reconC{\lst{\textrm{``loc''} \; \ell \; \vDF}} &= \ell \\
\textrm{otherwise, } \reconC{\vDF} &= \unquo{\vDF} 
\end{align*}
%
% reconS
%
$$\reconSname : \vDFCM_1,\ldots,\vDFCM_n \rightarrow \store$$
$$\reconS{\cons{\lst{s \; \ell \; {\vDF}'}}{\vDF},\vDF_{rest},\ldots} = $$
$$
\hspace{0.7cm} \left\{ \begin{array}{r}
 \cons{\pair{\ell}{\unquo{{\vDF}'}}}{\reconS{\vDF_{rest},\ldots}} \\
  \mbox{ if $\ell \notin \dom{\reconS{\vDF_{rest},\ldots}}$} \\
  \reconS{\vDF_{rest},\ldots} \hspace{3cm} \\
  \mbox{ if $\ell \in \dom{\reconS{\vDF_{rest},\ldots}}$} \\
 \end{array} \right.
$$
\begin{align*}
&\reconS{\vDF,\vDF_{rest},\ldots} = \reconS{\vDF_{rest},\ldots} \\
&\car{\vDF} \neq \lst{s \; \ell \; {\vDF}'}
\end{align*}
  \caption{Reconstruction function for \delayname and \forcename.}
  \label{fig:recon}
\end{figure}

Annotation adds \outname expressions and continuation mark \wcmname and
\ccmname operations to a program. For example, annotating the program
$\prim{+}{1}{2}$ results in the following annotated program:\footnote{For
  clarity, the syntactic sugar \texttt{let*} and \texttt{list} forms are
  used. They are defined in the standard way. Other minor code-readability
  improvements have also been made.}
\begin{alltt}

    \((\)let* 
      \((\)[\(t0\) \(\out{\cons{\quo{\prim{+}{1}{2}}}{\ccm}}\)]
       [\(v1\) \(\prim{+}{1}{2}\)]
       [\(t1\) \(\out{\cons{\quo{v1}}{\ccm}}\)]\()\)
     \(v1)\)
\end{alltt}
Annotated programs utilize the quoting function $\quoname$, which converts an
expression into a value representation. For example $\quo{\prim{+}{1}{2}} =
\lst{\texttt{"+"} \; 1 \; 2}$. There is also an inverse function, $\unquoname$,
for reconstruction. The above annotated program evaluates to 3, outputting the
values $\quo{\prim{+}{1}{2}}$ and $\quo{3}$ in the process, from which the
reduction sequence $\prim{+}{1}{2} \rightarrow 3$ can be recovered. The $\ccm$
calls in the example return the empty list since no continuation marks were
previously stored, i.e., there were no calls to \wcmname. There is no need for
\wcmname annotations because the entire program is a redex, i.e., the context
is empty.

Extending the example to $\prim{+}{\prim{+}{1}{2}}{5}$ yields
\begin{alltt}
  \((\)let*
    \((\)[\(v0\) (wcm \(\lst{\texttt{"prim2-1"} \texttt{"+"} 5}\)
              \((\)let* 
                \((\)[\(t0\) \(\out{\cons{\quo{\prim{+}{1}{2}}}{\ccm}}\)]
                 [\(v1\) \(\prim{+}{1}{2}\)]
                 [\(t1\) \(\out{\cons{\quo{v1}}{\ccm}}\)]\()\)
                \(v1))\)]
     [\(v2\) \(\prim{+}{v0}{5}\)]
     [\(t2\) \(\out{\cons{\quo{v2}}{\ccm}}\)]\()\)
   \(v2)\)
\end{alltt}
This extended example contains the first program as a subexpression; and
therefore the annotated version of the program contains the annotated
version of the first example. The $\prim{+}{1}{2}$ expression now occurs in
the context $\prim{+}{\holeE}{5}$ and the \wcmname expression stores an appropriate
continuation mark so this context can be reconstructed. The
\texttt{"prim2-1"} label indicates that the hole is in the left argument
position. The first \outname expression now produces the output value
$\lst{\quo{\prim{+}{1}{2}} \; \lst{\texttt{"prim2-1"} \; \texttt{"+"} \;
5}}$, which can be reconstructed to the expression
$\prim{+}{\prim{+}{1}{2}}{5}$. Reconstructing all outputs produces
$\prim{+}{\prim{+}{1}{2}}{5} \rightarrow \prim{+}{3}{5} \rightarrow 8$.

Storing context information in continuation marks also enables the
reconstruction of a machine state, which is what the stepper actually does, and
the reconstruction and annotation functions defined in
figures~\ref{fig:anno}~and~\ref{fig:recon} demonstrate how this works.
Figure~\ref{fig:anno} shows that if the subexpression $\eDF$ of a {\tt force}
expression does not evaluate to location, the annotations are like those for
the above examples. If $\eDF$ produces a location, additional continuation
marks (figure~\ref{fig:anno}, boxes 1 and 2) are needed to indicate the
presence of \forcename contexts during evaluation of a delayed computation. An
additional \outname expression (box 3) is also needed so that the steps showing
the removal of both the $\force{\holeE}$ and $\forceEDFdef{\holeE}$ contexts
can be reconstructed. Note the $\cdr{\ccm}$ (box 5) in the first \outname; this
ensures the $\forceEDFdef{\holeE}$ context is not part of the reconstructed
control string. The location $v0$ (box 4) is included in the output so the
store can be properly reconstructed. The \texttt{"val"} tag directs the
reconstruction function to use the value $\quo{v2}$ from the emitted
location-value pair for reconstructing the control string.

The annotation of a \delayname expression requires predicting the location
 of the delayed compuation in the store. We therefore assume we have access
 to an \allocname function that uses the same location-allocating algorithm
 as the memory management system of the machine.\footnote{Since labels are
 not displayed as numbers but as sharing among expressions, this
 unrealistic mathematical assumption is acceptable.} In addition to the
 location, the delayed expression itself (box 6) is also included in the
 output, to enable reconstruction of the store. The \texttt{"loc"} tag
 directs the reconstruction function to use the location from the emitted
 location-value pair for reconstructing the control string.

The reconstruction function in figure~\ref{fig:recon} consumes a list of
values, where each value is a sublist, and reconstructs a CS machine state from
each sublist. Again, only the cases involving \delayname and \forcename are
defined. The rest of the function is borrowed from
\ncite{Clements2006Thesis}. The first element of every $\vDFCM_i$ sublist
represents a (quoted) expression that is plugged into the context represented
by the rest of the sublist. The store is reconstructed by retrieving all the
location-value pairs in all the sublists up to the current one. The arguments
to the store-reconstruction function $\reconSname$ may contain duplicate
entries for a location, so a location-value pair is only included in the
resulting store if it does not occur in any subsequent arguments.

\begin{proof}[Lemma~\ref{lemma:annorecon} Proof Sketch]
The proof of lemma~\ref{lemma:annorecon} extends Clements's
 proof~\cite[Section 3.4]{Clements2006Thesis} with cases for \delayname and
 \forcename.  Also we modify the argumentation for the original cases to
 cope with the additional store, but doing so is straightforward because
 the store is simply threaded through. 
\end{proof}

%% -----------------------------------------------------------------------------
\subsection{Lazy Racket Correctness}

The function $\macroname : e \rightarrow \eDF$ macro-expands a Lazy Racket
program. Because source terms don't include labels, $\macroname$ is undefined
for labeled terms. Its partial inverse $\unmacroname : \CDF \times \store
\rightarrow e$ synthesizes an unlabeled Lazy Racket program from a (CS machine
representation of a) plain Racket program. The expansion and synthesis
functions are defined in figures~\ref{fig:macro} and \ref{fig:unmacro},
respectively.

\begin{figure}[htb]
\begin{align*}
  \macroname : e &\rightarrow \eDF \\
  \macro{\lam{x}{e}}     &= \lam{x}{\macro{e}} \\
  \macro{\app{e_1}{e_2}}  
    &= \app{\force{\macro{e_1}}}{\delay{\macro{e_2}}} \\
  \macro{\prim{p^2}{e_1}{e_2}} 
    &= \prim{p^2}{\force{\macro{e_1}}}{\force{\macro{e_2}}} \\
  \macro{\cons{e_1}{e_2}} 
    &= \cons{\delay{\macro{e_1}}}{\delay{\macro{e_2}}} \\
  \macro{\app{p^1}{e}}    &= \app{p^1}{\force{\macro{e}}} \\
  \macro{\iif{e_1}{e_2}{e_3}} &=
  \iif{\force{\macro{e_1}}}{\macro{e_2}}{\macro{e_3}} \\
  \textrm{otherwise, } \macro{e} &= e
\end{align*}
  \caption{Macro-expanding Lazy Racket.}
  \label{fig:macro}
\end{figure}

\begin{figure}
\begin{align*}
  \unmacroname : \CDF \times \store \rightarrow&\; e \\
  \unmacrodef{\CDF} =&\; e \\
    \textrm{where } &\cs{e}{\store'} = \unmacrotwodef{\CDF} \\
  \\
  \unmacrotwoname : \CDF \times \store \rightarrow&\; \cs{e}{\store} \\
% lambda
   \unmacrotwodef{\lam{x}{\cDF}} =& \csdef{\lam{x}{e}} \\
     \textrm{where } &\csdef{e} = \unmacrotwodef{\cDF} \\[3pt]
% app
   \unmacrotwodef{\app{\CDF_1}{\CDF_2}}
     =& \cs{\app{e_1}{e_2}}{\store''} \\
     \textrm{where } &\cs{e_1}{\store'} = \unmacrotwodef{\CDF_1} \\
                     &\cs{e_2}{\store''} = \unmacrotwo{\CDF_2}{\store'} \\
% prim
   \unmacrotwodef{\prim{p^2}{\CDF_1}{\CDF_2}}
     =& \cs{\prim{p^2}{e_1}{e_2}}{\store''} \\
     \textrm{where } &\cs{e_1}{\store'} = \unmacrotwodef{\CDF_1} \\
                     &\cs{e_2}{\store''} = \unmacrotwo{\CDF_2}{\store'} \\
% cons
   \unmacrotwodef{\cons{\CDF_1}{\CDF_2}}
     =& \cs{\cons{e_1}{e_2}}{\store''} \\
     \textrm{where } &\cs{e_1}{\store'} = \unmacrotwodef{\CDF_1} \\ 
                     &\cs{e_2}{\store''} = \unmacrotwo{\CDF_2}{\store'} \\
% p^1
   \unmacrotwodef{\app{p^1}{\CDF}}
     =& \cs{\app{p^1}{e}}{\store'} \\
     \textrm{where } &\cs{e}{\store'} = \unmacrotwodef{\CDF} \\[3pt]
% if
   \unmacrotwodef{\iif{\CDF_1}{\cDF_2}{\cDF_3}}
     =& \cs{\iif{e_1}{e_2}{e_3}}{\store'} \\
     \textrm{where } &\cs{e_1}{\store'} = \unmacrotwodef{\CDF_1} \\ 
                    &\cs{e_2}{\store'} = \unmacrotwo{\cDF_2}{\store'} \\
                   &\cs{e_3}{\store'} = \unmacrotwo{\cDF_3}{\store'} \\
% delay
   \unmacrotwodef{\delay{\cDF}} =&\; \unmacrotwodef{\cDF} \\
% locations
   \unmacrotwodef{\locdef} =&\; \unmacrotwodef{\storelookupdef} \\
% force
   \unmacrotwodef{\force{\CDF}} =&\; \unmacrotwodef{\CDF} \\
% forceloc
   \unmacrotwodef{\forceEDFdef{\CDF}}
     =& \cs{e}{\storeupdate{\store'}{\ell}{e}} \\
     \textrm{where } &\cs{e}{\store'} = \unmacrotwodef{\CDF} \\
% otherwise
 \textrm{otherwise, } \unmacrotwodef{\cDF} =& \csdef{\cDF}
\end{align*}
\caption{Synthesizing Lazy Racket from plain Racket.}
\label{fig:unmacro}
\end{figure}

Lemma~\ref{lemma:transition} states the correctness of Lazy Racket in terms of
the functions $\macroname$, $\unmacroname$, $\unlabelname$, and the \lambdaLR
rewriting system from section~\ref{sec:lambdalr}. That is, every CS machine
transition sequence has an equivalent \lambdaLR rewriting sequence, modulo
$\unmacroname$ and $\unlabelname$.

\begin{lemma}[LR Correctness]   \label{lemma:transition}
For all Lazy Racket programs $e$ and Racket programs \CDF such that
 $\csstart{\macro{e}} \cssteps \csdef{\CDF}$, there exists a Lazy Racket
 program $\eLR$ such that $e \multisteplr \eLR$ and 
 $\unmacrodef{\CDF} = \unlabel{\eLR}$.
\end{lemma}

\begin{proof}[Proof Sketch] We prove the lemma by induction on the number
 of CS machine steps. For the base case, the lemma holds because
 $\unmacro{\macro{e}}{\emptystore} = e$. Otherwise, we proceed by case
 analysis on the last transition step.

For each case, we prove correct synthesis of evaluation contexts and redexes
separately.  When the rewriting of a redex affects the context, too---see
figure~\ref{fig:lr}, third column---we need information about the redex for the
synthesis of a context. This parallel rewriting of \lambdaLR programs is
equivalent to the reduction of stored expressions in the CS machine if there
are multiple references to that expression. If the stored expression is a
value, the reconstruction naturally reifies the value throughout because
$\unmacroname$ translates locations by retrieving and unexpanding the
expression at that store location.

If a stored expression is in the process of being reduced to a value, however,
the store has not been updated but all references to this location must reflect
the partial reductions. Such intermediate steps manifest themselves as
reductions under $\forceEDFdef{\holeE}$ contexts in plain Racket. To synthesize
these intermediate steps properly, $\unmacroname$ first updates the store with
the partially reduced expression and then synthesizes the rest of the state
into a Lazy Racket program. As a result, the translation of all occurrences of
a location yield the desired intermediate expression. Only synthesis of
$\forceEDFdef{\CDF}$ control strings yields a new store; thus only
subexpressions that possibly contain $\forceEDFdef{\CDF}$, i.e., subexpressions
that can contain the redex, yield a new store.

For synthesis of a $\beta_v$ redex, the substituted argument is a location
 because $\macroname$ wraps all application arguments in a \delayname. The
 $\unmacroname$ synthesis function translates locations to the delayed
 expression at that location, so the corresponding \lambdaLR step is
 $\betalr$. The $\betalr$ step adds a label, which is removed with
 $\unlabelname$.

For synthesis of $\primcs$, $\carcs$, $\cdrcs$, $\iftruecs$, or $\iffalsecs$
redexes, $\unmacroname$ yields a $\primlr$, $\carlr$, $\cdrlr$, $\iftruelr$, or
$\iffalselr$ \lambdaLR redex, respectively, and these rewriting steps insert no
labels.

If the last CS step, $\csdef{\CDF} \csstep \cs{{\CDF}'}{\store'}$, follows
$\delaycs$, $\forcedelaycs$, $\forceupdatecs$, or $\forcenondelaycs$, there is
no corresponding \lambdaLR rewriting rule and $\unmacrodef{\CDF} =
\unmacro{{\CDF}'}{\store'}$. The lemma holds since this is equivalent to taking
zero steps.\end{proof}

\section{Performance} \label{sec:performance}

The performance of a stepper tool must be measured against the programmer's
ability to use the tool. In particular, raw performance numbers are
inconsequential because the usability of the tool primarily depends on its
responsiveness to I/O. 

Nevertheless, it is important to quantify the basic performance penalty of
 a stepper. In our case, the stepper slows down the execution of programs
 by a factor of 20 (up to 60), as the following table for a small number of
 micro-benchmarks shows:
\begin{center}
  \begin{tabular}{|c|c|}
    \hline
    \textbf{Test Name} & \textbf{Slowdown Factor} \\
    \hline
    \hline
    \texttt{fibo} & 21.4 \\
    \hline
    \texttt{ack}  & 32.7 \\
    \hline
    \texttt{partial} & 27.3 \\
    \hline
    \texttt{tak}  & 23.0 \\
    \hline
    \texttt{takl} & 34.9 \\
    \hline
    \texttt{takr} & 55.5 \\
    \hline
  \end{tabular}
\end{center}
 The numbers describe the performance ratio of annotated versus unannotated
 programs. While they include both annotation and evaluation time,
 the annotation time is negligible when compared to the time it takes to
 run the program.

We achieve adequate responsiveness of the stepper with a straightforward
arrangement. As soon as the stepper backend produces output, the front-end
asynchronously displays reduction steps. By the time the programmer has read
and understood the first step, the stepper can quickly respond to additional
requests.

\section{Related Work}

Debugging programs in lazy languages poses serious challenges, as numerous
attempts have shown over the past three decades.  Most recent debuggers are for
Haskell. \ncite{Allwood2009Needle} developed a stack-trace-generating tool for
GHC that aids in identifying the context of an error. Their tool transforms
programs to carry around an extra stack parameter for accessing the program
context, similar to continuation marks. Their tool approximates a call-by-value
stack, however, instead of displaying the lazy evaluation order.

\ncite{Ennals2003HsDebug} built HsDebug for GHC, which employs the ``stop,
examine, continue'' style found in imperative languages, e.g., gdb. With
HsDebug, programmers can set breakpoints and can examine the program state at
breakpoints. Like Allwood's debugger, HsDebug does not preserve laziness in
programs; the debugger must make certain tweaks, resulting in an ``optimistic''
evaluation model~\cite{Ennals2003Optimistic}.

The most recent version of the GHC system has debugging features built into its
interpreter~\cite{Marlow2007Debugger}. In principle, the GHCi debugger can most
closely mimic the operation of our stepper because: (1) it allows a programmer
to single-step through the evaluation of a program, and (2) it preserves
laziness for the user to observe. In contrast to our stepper, the GHCi debugger
does not use a substitution semantics, and only displays a few lines of the
program at a time, resulting in frequent jumps from the body of a function to
different call sites, when arguments need to be evaluated. This jumping can be
confusing for a programmer to follow, especially a novice to lazy
evaluation. Also, sharing is not easily observed in the GHCi debugger. All
thunks are rendered the same way (as a ``\_'' character), and when stepping,
the debugger skips over the reduction of a variable, so it is occasionally
unclear which argument is being evaluated.

Hat 2.0~\cite{Chitil2003Hat,Wallace2001NewHat} is a suite of debugging tools
for GHC that aggregates and improves on several previous tools: Hat
1.0~\cite{Sparud1997Hat1}, Hood~\cite{Gill2000Hood}, and
Freja~\cite{Nilsson1998Freja}. Hat's implementation resembles the
implementation of our stepper in that it works by transforming a source program
into an annotated one, which, when run, generates trace information. The
generated trace is then interpreted by the tools. In Hat, the generated trace
can be viewed in several different styles: as a directed graph of expressions
connected according to redex-contractum relations (like the original Hat), in
the question-answer style of an an algorithmic debugger (like Freja), or by
tracking specific values in a computation (like Hood). The graph viewer is
somewhat similar to our tool; however, program evaluation is portrayed using
graph semantics, and when viewing the graph, it can be difficult to deduce the
order of reductions for graphs larger than a few nodes and links. Hat also
includes a viewer, hat-stack, that shows a simulated eager evaluation stack,
similar to Allwood's tool. Hat is no longer maintained and it is not clear if
it still works with the more recent version of GHC.

Finally, only one of these tools come with formal models and correctness
 proofs for their architecture.  Chitil and Luo \cite{Chitil2006HatModel}
 developed a model for Hat's trace generator and show that the evaluation
 steps can be reconstructed from the information in the traces.

\section{Conclusion}

We have presented a lazy stepper as an additional tool for the lazy language
debugging arsenal. The stepper presents computation as the standard rewriting
sequence of a novel lazy semantics. While the stepper is implemented in Lazy
Racket via a ``macro'' over the existing stepper for strict Racket, our paper
explains the general software architecture via a generic theoretical model. We
conjecture that it is straightforward to construct a stepper on top of other
architectures, like Allwood's StackTrace.

% We recommend abbrvnat bibliography style.

\bibliographystyle{abbrvnat}

% The bibliography should be embedded for final submission.

\bibliography{lazystep}

%\begin{thebibliography}{}
%\softraggedright
%
%\bibitem[Smith et~al.(2009)Smith, Jones]{smith02}
%P. Q. Smith, and X. Y. Jones. ...reference text...
%
%\end{thebibliography}

\end{document}